%% file: BCSZ.tex
\documentclass[11pt]{article}

\usepackage{amsmath,amssymb}
\usepackage{latexsym}
\usepackage{graphicx}
\usepackage{color}
\usepackage{subfigure}
\usepackage{hyperref}
\usepackage{multirow}
\usepackage{url}
\usepackage{dsfont}

\newtheorem{theorem}{Theorem}[section]

\newtheorem{corollary}[theorem]{Corollary}
\newtheorem{definition}[theorem]{Definition}
\newtheorem{lemma}[theorem]{Lemma}

\newenvironment{proof}[1][Proof]{\textbf{#1.} }{\ \rule{0.5em}{0.5em} \vspace{1ex}}

\newcommand{\argmax}{\operatorname*{argmax}}
\newcommand{\argmin}{\operatorname*{argmin}}

\newcommand{\diag}{\operatorname{diag}}

\newcommand{\prob}{\operatorname{Pr}}

\newcommand{\rank}{\operatorname{rank}}
\newcommand{\tr}{\operatorname{tr}}

\newcommand{\np}[1]{\textsf{#1}}

\newcommand{\templen}{s} 

\begin{document}

\title{Multireference Alignment using Semidefinite Programming}

\author{
Afonso S. Bandeira\thanks{Program in Applied and Computational Mathematics (PACM),
Princeton University, Princeton, NJ 08544, USA ({\tt ajsb@math.princeton.edu}).}
\and
Moses Charikar\thanks{Department of Computer Science,
Princeton University, Princeton, NJ 08544, USA ({\tt moses@cs.princeton.edu}).}
\and
Amit Singer\thanks{Department of Mathematics and PACM,
Princeton University, Princeton, NJ 08544, USA ({\tt amits@math.princeton.edu}).}
\and
Andy Zhu\thanks{Department of Mathematics,
Princeton University, Princeton, NJ 08544, USA ({\tt azhu@math.princeton.edu}).}
}

\maketitle

\begin{abstract}
The multireference alignment problem consists of estimating a signal from multiple noisy shifted observations.
Inspired by existing Unique-Games approximation algorithms, we provide a semidefinite program (SDP) based relaxation which approximates the maximum likelihood estimator (MLE) for the multireference alignment problem.
Although we show that the MLE problem is Unique-Games hard to approximate within any constant, 
we observe that 
our poly-time approximation algorithm for the MLE appears to perform quite well in typical instances, outperforming existing methods. In an attempt to explain this behavior we provide stability guarantees for our SDP under a random noise model on the observations. This case is more challenging to analyze than traditional semi-random instances of Unique-Games: the noise model is on vertices of a graph and translates into dependent noise on the edges. 
Interestingly, we show that if certain positivity constraints in the SDP are dropped, its solution becomes equivalent to performing phase correlation, a popular method used for pairwise alignment in imaging applications.
Finally, we show how symmetry reduction techniques from matrix representation theory can simplify the analysis and computation of the SDP, greatly decreasing its computational cost.
\end{abstract}

\begin{center}
\end{center}

\newpage

\input{introduction_notation.tex}

\input{mle.tex}

\input{sdp.tex}

\input{stability.tex}


\input{numerics.tex}

\input{conclusion.tex}

\section*{Acknowledgements}
The authors thank Yutong Chen for valuable assistance with the implementation of our algorithm.
A.~S.~Bandeira was supported by AFOSR Grant No. FA9550-12-1-0317. 
M.~Charikar was supported by NSF grants CCF 0832797, AF 0916218 and AF 1218687.
A. Singer was partially supported by Award Number FA9550-12-1-0317 and FA9550-13-1-0076 from
AFOSR, by Award Number R01GM090200 from the NIGMS, and by Award Number LTR DTD 06-05-2012 from the Simons Foundation.
Parts of this work have appeared in A. Zhu's undergraduate thesis at Princeton University.

\bibliographystyle{alphaabbr}
\bibliography{thesisbib}

\appendix

\input{appendix_misc.tex}

\input{appendix_nphard.tex}

\input{appendix_staralgebra.tex}

\input{appendix_stability.tex}

\end{document}

%% file: introduction_notation.tex
\section{Introduction}


The multireference alignment problem is the following: suppose there is a template vector $x \in \mathbb{R}^L$, from which are sampled $N$ cyclically-shifted copies with white additive Gaussian noise \begin{align} 
y_i = R_{l_i}x + \xi_i \in \mathbb{R}^{L},\ \xi_i = (\xi_{ik})_{k=1}^L \sim \mathcal{N}(0, \sigma^2 I_L) \text{ i.i.d}. \label{notation_model} 
\end{align} for $i = 1, 2, \ldots, N$. $R_{l}$ is the index cyclic shift operator $(x_1, \ldots, x_L) \mapsto (x_{1-\ell}, \ldots, x_{L-\ell})$ and $\sigma$ is a parameter controlling the signal-to-noise (SNR) ratio. Both the template $x$ and the shifts $l_1,\ldots,l_n$ are unknown. Furthermore, no model is presumed a-priori for their distribution. From this model we would like to deduce an accurate estimate for $x$ (up to a global cyclic shift). Such an estimate can be obtained by first estimating the shifts, and then averaging the unshifted observations. For this reason we will focus on the problem of estimating the shifts $l_1,\ldots,l_n$.

This problem has a vast list of applications. Alignment is directly used in structural biology \cite{diamond92_superposition} 
\cite{theobald12_superposition}; radar 
\cite{zwart03_HRR} \cite{duch05_neural}; crystalline simulations \cite{sonday2011_sim}; and image registration in a number of important contexts, such as in geology, medicine, and paleontology \cite{dryden98_shape} \cite{foroosh02_subpixelregistration}. Various methods to solve this problem are used in these communities (see Appendix \ref{appendix_misc}).


A na\"ive approach to estimate the shifts in (\ref{notation_model}) would be to fix one of the observations, say $y_i$, as a reference template and align every other $y_j$ with it by the shift $\rho_{ij}$ minimizing their distance
\begin{equation}\label{eq:rhoij}
\rho_{ij} = \argmin_{l\in\mathbb{Z}_L}\|y_j-R_ly_i\|_2.
\end{equation}

This solution works well at a high signal-to-noise ratio (SNR), but performs poorly at low SNR. 
A more democratic approach would be to calculate all of the pairwise relative shift estimates $\rho_{ij}$ before attempting to recover the shifts $\{l_i\}$. This can be done by solving the minimization problem
\begin{equation}\label{eq:error_shiftspace}
\min_{l_1,\ldots,l_N\in \mathbb{Z}_L} \sum_{i,j=1}^N \left| e^{2\pi \imath l_i/L} - e^{2\pi \imath \rho_{ij}/L}e^{2\pi \imath l_j/L} \right|^2.
\end{equation}
This problem is known as angular synchronization~\cite{ASinger_2011_angsync,Bandeira_Singer_Spielman_OdCheeger} and the solution can be approximated via a SDP-based relaxation. 

The issue with attempting to solve the alignment problem using either (\ref{eq:rhoij}) or (\ref{eq:error_shiftspace}) is that one is evaluating the performance of a given choice of $l_i,l_j$ for a pair $(i,j)$ by how far $l_i-l_j$ is from $\rho_{ij}$, but not taking into account the cost associated with other possible relative shifts. Relating $R_{-l_i}y_i$ and $R_{-l_i}y_j$, for example, would take into account information about all possible shifts instead of just the best one. The quasi maximum likelihood estimator (Section \ref{section_mle}) attempts to do exactly that by solving the minimization problem:
\begin{equation}\label{eq:error_signalspace}
\min_{l_1,\ldots,l_N\in \mathbb{Z}_L} \sum_{i,j=1}^N \left\| R_{-l_i}y_i - R_{-l_j}y_j \right\|^2.
\end{equation}

Finding the MLE (\ref{eq:error_signalspace}) is a non-trivial computational task because the parameter space is of exponential size, and the likelihood function is non-convex. While one can apply optimization methods such as gradient descent, simulated annealing, and expectation-maximization (EM), these are only guaranteed to find local minima of (\ref{eq:error_signalspace}), but not the global minimum.



In this paper we take a different approach and propose a semidefinite relaxation for the quasi maximum likelihood problem (\ref{eq:error_signalspace}). This particular SDP is inspired by an approximation algorithm designed to solve the Unique Games problem~\cite{Charikar_Makarychev_UG} (Section~\ref{section_sdp}). 


Convex relaxations of hard combinatorial problems have seen many successes in applied mathematics. They became particularly popular in the last decade with the introduction of Compressed Sensing, in the seminal work of Donoho, Candes, Tao, and others~\cite{Candes_CS2,Donoho_CS}. 
This idea has since been applied to a vast list of problems. Semidefinite programming (SDP) has served as a convex surrogate for problems arising in applications such as low-rank matrix completion~\cite{Candes_MatrixCompletion}, 
phase retrieval~\cite{Candes_Strohmer_Voroninski_phaselift}, Robust PCA~\cite{Tropp_RobustPCA}, 
multiple-input multiple-output (MIMO) channel detection~\cite{So_MIMO}, and many others. In many of these applications the same phenomenon is present: for typical instances, solving the convex problem is often equivalent to solving the original combinatorial problem \cite{Tropp_livingontheedge}.

Convex relaxations (and, in particular SDP based relaxations) also play a central role in the design of approximation algorithms in theoretical computer science. Almost two decades ago, Goemans and Williamson~\cite{MXGoemans_DPWilliamson_1995} proposed a SDP based approximation algorithm for the \np{MAX-CUT} problem with approximation ratio $\alpha^{\np{GW}} \approx 0.878$. That is, for any instance of the problem, the computed solution is guaranteed to provide performance (in this case, a cut) at least $\alpha^{\np{GW}}$ of the optimum. Many semidefinite relaxations have since been proposed as approximation algorithms for a long list of NP-hard problems~\cite{Design_ApAlg}. 

In order to better understand the theoretical limitations of approximation algorithms, substantial work has been done to establish limits on the approximation ratios achievable by poly-time algorithms for certain \np{NP}-hard problems 
(hardness of approximation).
The Unique Games conjecture by Khot~\cite{SKhot_2002} is central to many recent developments:
For $\delta, \varepsilon > 0$, it is impossible for a polynomial-time algorithm to distinguish between $\delta$-satisfiable and $(1 - \varepsilon)$-satisfiable \np{Unique-Games} instances.
A \textsf{Unique-Games} instance consists of a graph along with a permutations for each edge. The problem is to choose the best assignment of labels to each vertex such that as many of the edge permutations are satisfied. The validity of the UGC would imply the optimality of certain poly-time approximation ratios, in particular the Goemans-Williamson constant $\alpha_{\textsf{GW}}$ for the \textsf{MAX-CUT} problem \cite{MXGoemans_DPWilliamson_1995,KKMO}. 



The best known polynomial time approximation to the unique games problem~\cite{Charikar_Makarychev_UG} is based on an SDP relaxation of a formulation that uses indicator variables. It is quite different from SDPs normally used in applications (such as those described above). In particular, the variable matrix has size $NL\times NL$ and $\Omega\left(N^2L^2\right)$ constraints. We adapt this SDP to approximate the quasi maximum likelihood problem (\ref{eq:error_signalspace}).


As we show that it is Unique-Games hard to approximate (\ref{eq:error_signalspace}) within any constant (see Section~\ref{section_mle}) it is hopeless to aim for good guarantees for general instances. 
However worst case analysis is often too pessimistic and not indicative of performance 
observed in practice.
In fact, under the random noise model we have for the observations, numerical simulations suggest that the SDP relaxation performs remarkably well, seeming to outperform existing methods. In an attempt to explain this phenomenon we show that our SDP is stable at high SNR levels and that it is tight at extremely high SNR levels. By stability, we mean that with high probability the solution to the SDP does not deviate much from the true solution (see Theorem \ref{stability_theorem}). The stability for this SDP is particularly interesting as it is more challenging to analyze than random instances of Unique-Games~\cite{Arora_UGexpanders,Kolla_Tulsiani_UG}, since our noise model is on the vertices, which translates into dependent noise on the edges.
Still, these results fall short of properly explaining the remarkable performance that we see in simulations and more research is needed towards understanding the typical behavior of this SDP.

In order to simplify the SDP we also study a version with fewer constraints. 
Interestingly, this weaker SDP can be solved explicitly and is equivalent to the pairwise alignment method called phase correlation~\cite{Horner_PhaseCorrelation}. This method does not take into account information between all pairs of measurements, which suggests that the full complexity of the Unique-Games SDP~\cite{Charikar_Makarychev_UG} is needed to obtain a good approximation to (\ref{eq:error_signalspace}). 

The fact that a global shift does not affect the solution to (\ref{eq:error_signalspace}) creates symmetries in our SDP relaxation. In fact, we leverage such structure by using symmetry reduction techniques from matrix representation theory to simplify the analysis and computation of the SDP, greatly decreasing its computational cost. This is quite useful given the high computational cost of semidefinite programming.


{\bf Contributions:} Our main contribution is applying techniques from theoretical
Computer Science to a problem in applied Math. We introduce an Unique Games style
SDP relaxation for the alignment problem that is novel for the applied Math community.
From the theoretical Computer Science point of view, we introduce a new problem that
has a similar flavor to the Unique Games problem -- in fact we show that the worst case
version is at least as hard as Unique Games. We introduce a natural
average case version of this alignment problem - aligning several shifted copies of a signal
corrupted by independent Gaussian noise. Existing analyses of semi-random models of Unique Games
do not seem to apply to this problem. We show that for sufficiently high SNR, the SDP
solution is close to an integer solution -- this is a first step to establishing a signal recovery
result which we leave as an open problem. We believe that future investigations into this
problem will yield interesting insights into the Unique Games SDP and on dealing with
correlated noise in average case analysis.

%% file: mle.tex
\section{Quasi Maximum Likelihood Estimator} 
\label{section_mle}

The log likelihood function for the model (\ref{notation_model}) is \begin{align}
\mathcal{L}(x, l_1, \ldots, l_N \,|\, y_1, \ldots, y_N) = \frac{N}{2} \log (2\pi) - \frac{1}{2\sigma} \sum_{i \in [N]} \|R_{-l_i}y_i - x\|^2. \label{mle_log_likelihood}
\end{align} Maximizing $\mathcal{L}$ is equivalent to minimizing the sum of squared residuals $\sum \|R_{-l_i}y_i - x\|^2.$ Fixing the $l_i$'s, the minimal value of $\mathcal{L}$ occurs at the average $x = \frac 1N \sum_{i=1}^N R_{-l_i} y_i$. Making the tame assumption that $\|x\|^2$ is estimable (the norm is shift-invariant), maximizing (\ref{mle_log_likelihood}) is thus equivalent to maximizing the sum of the inner products $\left\langle R_{-l_i}y_i,\, R_{-l_j}y_j \right\rangle$ across all pairs $(i,j)$. 
%
Thus we consider the estimator \begin{align}
\widehat{\boldsymbol{\ell}} = \argmax_{\boldsymbol{\ell} \in \mathbb{Z}_L^N} \sum_{i,j \in [N]} \langle R_{-l_i} y_i, R_{-l_j} y_j \rangle
\label{mle_argmax}
\end{align} 
Unfortunately, the search space for this optimization problem has exponential size. Indeed, assuming no model for the vectors $\{y_i\}$, it is \np{NP}-hard to find the shifts which maximize (\ref{mle_argmax}), or even estimate it within a close constant factor.
\begin{theorem}
It is \np{NP}-hard (under randomized reductions) to find a set of labels approximating (\ref{mle_argmax}) within $16/17+\varepsilon$ of its optimum value. It is \np{UG}-hard (under randomized reductions) to approximate (\ref{mle_argmax}) within any constant factor.
\end{theorem}
\begin{proof}
(outline)
We give a randomized reduction from the class of $\Gamma$\np{-MAX-2LIN}($q$) instances consisting of a set of $2$ variable linear equations of the form $x_i - x_j \equiv c_{ij} \pmod{q}$, with the goal of choosing an assignment for the variables which maximizes the number of satisfied equations.
We construct a vector $y_k$ for every variable $x_k$ such that shifts of
$y_k$ correspond to an assignment to $x_k$.
We pick a random vector $z_{ij}$ corresponding to a constraint on 
variables $x_i,x_j$ and place a copy of $z_{ij}$ at specific locations in $y_i$ and $y_j$.
Shifts of $y_i,y_j$ corresponding to satisfying assignments of
the constraint results in a superposition of the copies of $z_{ij}$.
We choose parameters so that the only non-trivial contributions to
the objective function (\ref{mle_argmax}) come from such superposition.
The value of the objective is (within small error) a scaled version of the 
number of constraints of the $\Gamma$\np{-MAX-2LIN}($q$) instance
satisfied by the assignment corresponding to the shifts.
Thus hardness results for $\Gamma$\np{-MAX-2LIN}($q$) directly 
translate to hardness results for the alignment problem.
The details are given in Appendix \ref{section_appendix_nphard}.
\end{proof}

The discrete optimization problem (\ref{mle_argmax}) may be formulated using indicator variables as an integer programming problem \begin{align}
\label{sdp_integerprogramming}
\argmax_{\{u_{ik}\}}\ \sum_{i,j=1}^N \sum_{k,l \in \mathbb{Z}_L} u_{ik}u_{jl} \langle R_{-k} y_i, R_{-l} y_j \rangle,
\end{align} where $u_{ik} \in \{0,1\}$ is the indicator variable $u_{ik} = \delta\{l_i \equiv k\}$. 
View $U_{ik;jl} = u_{ik}u_{jl}$ as an entry of the Gram matrix $U \in \mathbb{R}^{NL \times NL}$, and $C_{ik;jl} = \langle R_{-k} y_i, R_{-l} y_j \rangle$ as an entry of the data Gram matrix $C \in \mathbb{R}^{NL \times NL}$. (\ref{sdp_integerprogramming}) can be written as the maximization problem $\tr(CU)$ subject to $\rank(U) = 1$, and any other constraints needed to enforce that $U$ is a Gram matrix of indicator variables. This suggests the possibility of a spectral rounding algorithm, where we try to determine the indicator variables by examining the top eigenvectors of $C$.

\begin{lemma}
\label{mle_gram}
The data Gram matrix $C$ with entries $C_{ik;jl} = \left\langle R_{-k} y_i, R_{-l} y_j \right\rangle$ satisfies:
\begin{enumerate}
\item $C \succeq 0$ and has rank $L$, with non-zero eigenvalues $\lambda_k = L \sum_{i=1}^N \left|\mathcal{F}(Y_i,k)\right|^2.$
\item There is a unitary matrix $\mathcal{P}$ for which $\mathcal{P}C\mathcal{P}^* = \diag(\mathcal{C}_0, \ldots, \mathcal{C}_{L-1})$ is block diagonal, where each $\mathcal{C}_k \in \mathbb{C}^{N \times N}$ is a rank $1$ matrix.
\end{enumerate}
\end{lemma}
\begin{proof}
Refer to Lemma \ref{appendix_mle_gram}.
\end{proof}

For random signals $x$, Lemma \ref{mle_gram} indicates that the spectral gap between the top $L$ eigenvalues and the remaining eigenvalues of $C$ will be large, on the order of $\min_k \sum_{i} |\mathcal{F}(Y_i,k)|^2  = \Omega(LN)$. Hence we may want to try and recover a solution to \ref{sdp_integerprogramming} by examining the eigenvectors associated with the top eigenvalues in the spectrum of $C$. In particular, in the noiseless case with $\sigma = 0$, the indicator vector $\mathds{1}\{ik: k \equiv l_i\} \in \mathbb{R}^{NL}$ lies in the span of the top $L$ eigenvectors of $C$. 

Unfortunately, this spectral gap will not be apparent for a large class of signals. As long as a single power spectra $|\mathcal{F}(x,k)|^2$ of $x$ is near zero, the corresponding eigenvalue $\lambda_k$ will separate less from the small eigenvalues of $C$, and hence the space of the top $L$ eigenvectors of $C$ will contain less relevant information. For this reason, it would only be meaningful for us to characterize the SNR by the spectral gap $\min_k |\mathcal{F}(x, k)|^2$. Furthermore, our simulations suggest that recovery from a spectral relaxation performs worse than the semidefinite relaxation we are about to propose.

%% file: sdp.tex
\section{Semidefinite relaxation}
\label{section_sdp}

The quadratic form constraints  
\begin{align}
\sum_{k,l \in \mathbb{Z}_L} u_{ik}u_{jl} = 1, &\quad i, j \in [N] \notag \\
u_{ik}u_{il} = 0, &\quad i \in [N],\ k \neq l  \in \mathbb{Z}_L \notag \\
u_{ik}u_{jl} \ge 0, &\quad i, j \in [N],\ k, l \in \mathbb{Z}_L. \notag
\end{align}
on $u_{ik} \in \mathbb{R}$ enforce that the $u_{ik}$'s are indicator variables (up to global sign, which cannot be fixed by quadratic constraints). The global shift ambiguity of the problem guarantees the existence of $L$ different solutions. In fact, we will attempt to find a candidate solution in the span of lifted versions of these.
This can be more succintly written as an optimization problem for a rank $L$ matrix $V$: \begin{align}
\max_V &\quad \tr(CV) \label{sdp_rankoneformulation} \\
\text{subject to} 
&\quad C,V \in \mathbb{R}^{NL \times NL},\ C_{ik;jl} = \langle R_{-k}y_i, R_{-l}y_j \rangle \notag \\
&\quad \sum_{k,l} V_{ik;jl} = 1,\ V_{ik;il} = 0,\ V \ge 0 \text{ for } k \neq l,\ V \succeq 0. \notag 
\end{align} 
Of the imposed constraints on $V$, the only non-convex constraint is that of rank deficiency, which obstructs the use of convex programming techniques. Removing this rank constraint yields a semidefinite program.
This SDP is extremely similar to (and motivated by) SDPs commonly used to approximate solutions to certain constraint satisfaction problems (\textsf{CSP}s), notably \textsf{Unique-Games} 
instances.
An \textsf{Unique-Games} instance consists of a graph $G = ([N], E)$, a label set $\mathbb{Z}_L$, and a set of permutations $\pi_{ij} : \mathbb{Z}_L \to \mathbb{Z}_L$. The problem is to choose the best assignment of labels to each vertex such that as many of the permutations $(\pi_{ij})_{(i,j) \in E}$ are satisfied. \textsf{$\Gamma$-MAX-2LIN($L$)} is special case where the permutations $\pi_{ij}$ are cyclic. 
In our notation, the SDP studied for \textsf{Unique-Games} is usually of the form
\begin{align}
\max_V \quad & \frac 12 \tr(CV) \label{sdp_uniquegames} \\ 
\text{subject to} \quad & C,V \in \mathbb{R}^{NL \times NL},\ C_{ik;jl} = \delta\{l = \pi_{ij}(k)\}, \notag \\
& \sum_{k} V_{ik;ik} = 1, \ V_{ik;il} = 0 \text{ for } k \neq l,\ V \ge 0, \ V \succeq 0. \notag 
\end{align} 
This formulation attempts to count the number of satisfied edge constraints for an \textsf{Unique-Games} instance.
One can treat the SDP (\ref{sdp_rankoneformulation}) as an instance of \textsf{$\Gamma$-MAX-2LIN($L$)} on a weighted complete graph, with each cyclic permutation weighted by $C_{ik;jl} = \langle R_{-l_i}y_i, R_{-l_j}y_j \rangle$. 
In this context, the matrix $C$ is dubbed the label-extended adjacency matrix \cite{kolla10_spectralUG}. Thus the significant body of literature conducted on \textsf{Unique-Games} may be useful in understanding (\ref{sdp_rankoneformulation}). Another common feature of the aligment SDP with
$\Gamma$\textsf{-MAX-2LIN($L$)} instances is that the assigned labels may be chosen upto cyclic symmetry. This induces a block circulant symmetry in the semidefinite program (see Lemma \ref{appendix_staralgebra_groupinvariant}). For example, this symmetry will allow us to presume that $V_{ik;ik} = 1/L$ from the constraint $\sum_k V_{ik;ik} = 1$.

A major feature of (\ref{sdp_rankoneformulation}) which does not feature in the study of \textsf{Unique-Games} is the structure of the data coefficient matrix $C$. While \textsf{Unique-Games} specifies constraints on edges of a graph (there are $N^2$ pieces of information), the alignment problem only specifies information on its vertices ($N$ pieces of information). This does assist our understanding of the semidefinite program, since it enables us to apply more symmetry conditions, but it also significantly complicates our analysis. 

An interpretation of the constraint $V \ge 0$ is that it enforces triangle inequality constraints $\|v_{ik} - \bold{0}\| + \|\bold{0} - v_{jl}\| \ge \|v_{ik} - v_{jl}\|$ \cite{Charikar_Makarychev_UG}. 
As we see in the next section, as well as from empirical results, the constraint $V \ge 0$ causes the SDP solution matrix to stabilize more around integral instances.
Interestingly, without this positivity constraint, the SDP may be solved in closed form, and is effectively equivalent to pairwise phase correlation. 

\begin{theorem}
\label{sdp_phasecorrelationequivalence}
Let V be a solution of (\ref{sdp_uniquegames}) without the positivity constraints. Then, $V$ has rank $L$, corresponding to one eigenspace of eigenvalue $N/L$. This eigenspace contains the vector $v^{\mathsf{phase}} \in \mathbb{R}^{NL}$ satisfying \[
\left\{v^{\mathsf{phase}}_{il}\right\}_l = \mathcal{F}^*\left\{ \frac{\mathcal{F}(y_1,l) \mathcal{F}^*(y_i,l)}{|\mathcal{F}(y_1,l) \mathcal{F}^*(y_i,l)|}\right\}_l
\] which is the concatenation of phase correlation vectors (see Appendix \ref{appendix_misc}) between $y_1$ and $y_i$.
\end{theorem}
\begin{proof}
By symmetry, assume without loss of generality that $V$ is block circulant (see Lemma \ref{appendix_staralgebra_groupinvariant}). Let $\mathcal{P}$ be the unitary matrix defined in Lemma \ref{mle_gram} and $e\left(\tfrac{kl}L\right)$ denote the classical Fourier basis function $e(x) = e^{2\pi i x}$. Then, $\mathcal{P}V\mathcal{P}^* = \diag(\mathcal{V}_0, \ldots, \mathcal{V}_{L-1})$ where $
\mathcal{V}_k = \sum_{l = 0}^{L-1} e\left(\tfrac{kl}L\right) \, V_l \in \mathbb{C}^{N \times N}
$ is positive semidefinite. The SDP constraints \[
V \in \mathbb{R}_{NL \times NL}, \quad V \succeq 0, \quad V_{ik;il} = 0 \text{ for } k \neq l, \quad \sum_{k,l} V_{ik;jl} = 1
\] are respectively equivalent to the Fourier side constraints 
\begin{align}
\label{sdp_fourierconstraints}
\mathcal{V}_k = \overline{\mathcal{V}_{L-k}}, \quad \mathcal{V}_k \succeq 0, \quad (\mathcal{V}_k)_{ii} = 1/L, \quad (\mathcal{V}_0)_{ij} = 1/L.
\end{align}
Since $\mathcal{V}_k \succeq 0$, the absolute value of each entry is bounded by the maximum of its diagonal entries, so each entry has magnitude at most $1/L$. Hence \begin{align}
\tr(CV) &= \sum_{k \in \mathbb{Z}_L} \tr(\mathcal{C}_k\mathcal{V}_k) \le \frac 1L \sum_{k \in \mathbb{Z}_L} \sum_{ij \in [N]} |(\mathcal{C}_k)_{ij}|,
\end{align} 
with equality occuring when $\mathcal{V}_k$ has the entries of $\mathcal{C}_k$, normalized to magnitude $1/L$. Then $(\mathcal{V}_k)_{ij} = \frac{\mathcal{F}(y_i,l) \mathcal{F}^*(y_j,l)}{|\mathcal{F}(y_i,l) \mathcal{F}^*(y_j,l)|}
$, and a basis of the non-trivial eigenspace of $V$ is given by \[
W := \mathcal{P}^* \left((\mathcal{V}_k)_{i,1}\right)_{i,k} =
\left(e(\tfrac{-kl}L) (\mathcal{V}_k)_{i,1}\right)_{ik;l}
\] Note the selection of the observation $1$ is arbitrary. The vector $v^{\mathsf{phase}} \in \mathbb{R}^{NL}$ given by \[
v^{\mathsf{phase}}_{ik} := (W \bold{1}_{N \times 1})_{ik} 
= \left[\mathcal{F}^*\left\{ \frac{\mathcal{F}(y_1,l) \mathcal{F}^*(y_i,l)}{|\mathcal{F}(y_1,l) \mathcal{F}^*(y_i,l)|} \right\}_{l=0}^{L-1} \right]_k
\] lies in the image of $W$ and it is easy to see that the $L$ vectors $W^{\mathsf{phase}} \in \mathbb{R}^{NL \times L}$ generated by circulating the $N$ blocks of $L$ entries of $v^{\mathsf{phase}}$ are linearly independent. Indeed, the top $L \times L$ block of $W^{\mathsf{phase}}$ is the identity matrix $I_L$. 
\end{proof}


From the SDP solution, one must round back to a solution in the original search space. There is a significant body of literature on the topic of rounding the solutions to various SDPs for Unique Games. 
The analysis and guarantees for these rounding schemes are in terms of the
number of constraints satisfied by the solution produced and do not
immediately give a result about signal recovery in our setting.
In their study of semi-random instances of \np{Unique-Games}, the authors of \cite{Kolla_Makarychev2_UG} give a rounding technique that uses both SDP and LP solutions when the SDP is known to be somewhat sparse -- this is similar to a condition we obtain in the following stability section.
Exploiting these ideas to establish an exact signal recovery guarantee is an interesting
open problem.

%% file: stability.tex
\section{Stability}

For simplicity, without loss of generality align the ground truth shifts of the observations, so that $y_i = x + \xi_i$, where $\xi_i \sim \mathcal{N}(0, \sigma^2 I_{L})$ i.i.d. The ground truth integral instance for the SDP will be (upto block cyclic symmetry) the indicator matrix $V^{int} \in \mathbb{R}^{NL \times NL}$, defined as $V_{ik;jl}^{int} = \delta_{k=l}/L$. If there is a set of labels $\boldsymbol{\ell} = (l_i)_i$ which are pairwise optimal, in the sense that $\langle R_{-l_i}y_i, R_{-l_j}y_j \rangle = \max_{k} \langle y_i, R_{-k} y_j \rangle$ for all pairs $(i,j)$, the SDP (\ref{sdp_rankoneformulation}) will produce the integral instance. While this may be likely in a very high signal-to-noise setting, it is a rather stringent condition. We relax this condition and show that the SDP solution must still resemble the integral instance at a reasonably high SNR. The exact definition for the SNR will be deferred for later, but will be characterized in terms of the gap 
\begin{align}
\label{stability_delta}
\Delta = \|x\|^2 - \max_{l \neq 0} \langle x, R_l x \rangle. 
\end{align}

For any $V \in \mathbb{R}^{NL \times NL}$ lie in the SDP feasibility region, we can characterize the distance of $V$ from the integral instance by the differences $$D_{ij} = \sum_{k \neq l \in \mathbb{Z}_L} V_{ik;jl} = 1 - \sum_{k \in \mathbb{Z}_L} V_{ik;jk} \in [0,1].$$ $D_{ij}$ is a measure of how much the SDP would weight shift preferences other than the ground truth. When all $D_{ij} = 0$ we obtain the integral instance. 

For convenience, we make the definitions $\xi_0 = 2x$ and $\eta_{ij} = 2 \max_l |\langle R_l \xi_i, \xi_j \rangle| \ge 0$ for $i,j = 0, 1, \ldots, N$. The following lemma demonstrates that the $\eta_{ij}$'s can be treated as worst-case bounds on the noise terms of the entries of the data Gram matrix $C$.

\begin{lemma}
\label{appendix_stability_ineq1}
If $\tr(CV) \ge \tr(CV^{int})$, then $\sum_{i \neq j} (\Delta - \eta_{ij} - \eta_{j0}) D_{ij} \le 0$. 
\end{lemma} 
\begin{proof}
We can find a block circulant matrix which attains the same SDP objective value as $V$, so without loss of generality presume $V$ is block circulant. Hence $\sum_k V_{ik;j0} = 1/L$ and $D_{ij} = 1 - LV_{i0;j0}$. Expanding,
\begin{align*}
0 &\ge \tr(CV^{int}) - \tr(CV) = 
\sum_{i,j} \sum_{k} \langle R_{-k}y_i, y_j \rangle \left(\delta_{k = 0} - L V_{ik;j0} \right)
\\
&\ge \sum_{i,j} \bigg( \langle y_i, y_j \rangle (1 - L V_{i0;j0}) - \max_{l \neq 0} \langle R_ly_i, y_j \rangle \sum_{k \neq 0} L V_{ik;j0} \bigg)\\
&= \sum_{i,j} \left( \langle y_i, y_j \rangle - \max_{l \neq 0} \langle R_ly_i, y_j \rangle \right) D_{ij}.
\end{align*}
The second inequality requires the SDP constraint $V \ge 0$.
The pessimistic bound 
\begin{align*}
\langle y_i, y_j \rangle - \max_{l \neq 0} \langle R_ly_i, y_j \rangle &\ge \|x\|^2 + \langle x,\xi_i\rangle + \langle x,\xi_j \rangle + \langle \xi_i, \xi_j \rangle - \max_{l \neq 0} \langle x, R_l x \rangle  \\
&\qquad - \max_{l \neq 0} \langle R_lx,\xi_i\rangle - \max_{l \neq 0}  \langle R_lx,\xi_j \rangle - \max_{l \neq 0}  \langle R_l \xi_i, \xi_j \rangle  \\
&\ge \Delta - ( \eta_{i0}/2 +  \eta_{j0}/2 +  \eta_{ij}),
\end{align*}
and rearrangement gives the desired inequality. 
\end{proof}

\begin{theorem}
\label{stability_theorem}
With probability $1 - e^{-N + o(N)}$, the solution to the SDP satisfies $$
\sum_{i,j} D_{ij} \le \frac{(\|x\| + \sigma^2 \sqrt{L}) \cdot 12\log eL }{\Delta} \cdot N^2.
$$
\end{theorem}
\begin{proof}
For sufficiently high SNR, we would expect that the inequality in Lemma \ref{appendix_stability_ineq1} would fail to hold. Indeed, by Lemma \ref{stability_independenttailbounds}, the inequality
\begin{align}
&\sum_{i \neq j} (\eta_{ij} + \eta_{j0}) D_{ij} \le \mathcal{O}(\log L) \cdot \bigg(2\|x\| + \frac 1N\sum_i \|\xi_i\|\bigg) N^2,
\label{stability_independenttailbounds_ineq}
\end{align} holds with probability at least $1 - e^{-N + \mathcal{O}(\log N)}$. It arises from tail bounds on the sum of slightly dependent random variables, and is independent of the structure of the SDP (as opposed to Lemma \ref{appendix_stability_ineq1}). Combined with Lemma \ref{appendix_stability_ineq1}, we can obtain a guarantee on the deviation between the SDP solution and the integral instance. The full proof may be found following Theorem \ref{appendix_stability_theorem}.
\end{proof}

Theorem \ref{stability_theorem} indicates that at a sufficiently high SNR, the UG-based SDP will produce a matrix $V$, of which each $L \times L$ block will have high small values outside of the main diagonal. A rounding scheme would likely interpret this as the identity shift being the optimal shift. This motivates us to choose a definition for the SNR to be along the lines of $SNR = \Delta/[(\|x\| + \sigma^2 \sqrt{L})\log L]$ (for reference, for random signals $x$, we would expect that $\Delta = L - \mathcal{O}(\sqrt{L})$). 
This characterization of SNR is significantly more lenient than that for spectral relaxation, in Section \ref{section_mle}. For example, a sum of a few sinusoids will have several small power spectra, but if the least common multiple of their periods does not divide $L$, then $\Delta$ will still be large. Note furthermore that we may be more flexible in our definition for $\Delta$: for example, suppose there is a set of shifts $\boldsymbol{\ell}^*$ (including the identity element) of size $|\boldsymbol{\ell}^*| = \mathcal{O}(1)$ for which $\max_{l \in \boldsymbol{\ell}^*} \langle x, R_l x \rangle > \Delta + \min_{l \notin \boldsymbol{\ell}^*} \langle x, R_l x \rangle$. The above results may be modified to describe the concentration of the SDP solution on the entries of the shifts in $\boldsymbol{\ell}^*$ (not just along the identity shift) for each pair of observations $y_i, y_j$. 

This result may be strengthened in other manners. For constants $0 < \delta, \varepsilon \ll 1$, we can attain a tighter concentration condition of the form 
\begin{align}
\sum D_{ij} \ge (1+\delta)^2N\sqrt{\sum D_{ij}^2} / SNR + 2\varepsilon N^2 \label{stability_jl_concentration}
\end{align} instead of the current condition $\sum D_{ij} \ge N^2 / SNR$. The argument is based on the analysis of the SDP for adversarial semi-random \textsf{Unique-Games} instances by \cite{Kolla_Makarychev2_UG}. However, the proof must be more nuanced in our case, due to correlations in the noise model of $C$, caused by the smaller source of randomness available to us. Hence we omit the full details, but provide a sketch below.

Any SDP feasible matrix $V$ can alternatively be represented as $V_{ik;jl} = \langle v_{ik}, v_{jl} \rangle / L$ for a set of unit vectors $v_{ik} \in \mathbb{R}^{NL}$. With this notation, $2D_{ij} = 2(1 - L V_{i0;j0}) = \|v_{i0} - v_{j0}\|^2$. The space of these unit vectors can be approximately discretized by a random projection due to the Johnson-Lindenstrauss lemma. Precisely, Lemma \ref{appendix_stability_johnsonlindenstrauss} provides a set of unit vectors $\mathfrak{N}$ of size $|\mathfrak{N}| = \exp\left(\mathcal{O}\left(\delta^{-2} \log^2(1/\varepsilon) \right)\right)$ such that the set of unit vectors $\{v_{i0}\}_{i=1}^N$ can be approximated under a random projection $\varphi$ by an $N$-tuple in $\mathfrak{N}$. Specifically, if $D_{ij}^\varphi = \|\varphi(v_{i0}) - \varphi(v_{j0})\|^2/2$, then the $D_{ij}$'s lie within an $\alpha_{JL}$-ball of $D_{ij}^\varphi$, this ball being defined as $\{D: \alpha_{JL}^{-1}(D^\varphi) \le D \le \alpha_{JL}(D^\varphi)\}$ for $\alpha_{JL}(D) = (1+\delta)D + \varepsilon$.

Instead of finding a global tail bound in Lemma \ref{stability_independenttailbounds}, we can derive a local tail bound for each $\alpha_{JL}$-ball of $D_{ij}^{\varphi}$'s, each tail bound holding with probability at least $1 - 2N e^{-t^2 N}$. For $0 < \delta,\, \varepsilon = \varepsilon' \ll 1$ constants, the number of $N$-tuples of vectors in $\mathfrak{N}$ is of size $\exp(\mathcal{O}(N))$. With a sufficiently large constant $t$, the local tail bound may be union bounded across all balls of vectors in $\mathfrak{N}^N$, and thus will hold for all SDP-feasible $V$. With some care, this argument would yield a concentration condition of the form (\ref{stability_jl_concentration}).

%

%% file: numerics.tex
\section{Numerical Results}
\label{section_numerics}

\begin{figure}[h]
\centering
\label{numerics_plots}
\includegraphics[scale=0.5]{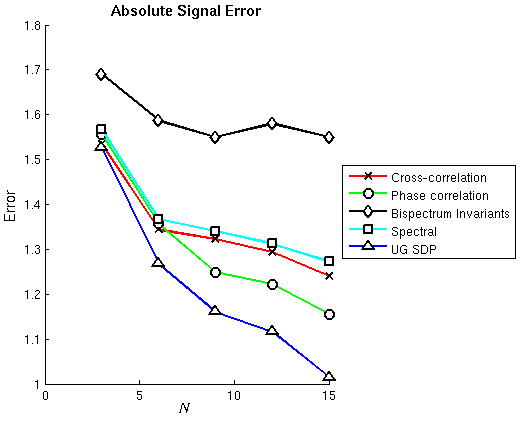}
\hspace{3em}
\includegraphics[scale=0.5]{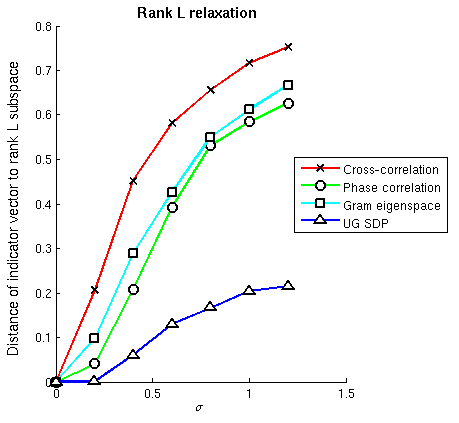}
\caption{Averages of errors of several alignment methods across 500 iterations, with parameters $\sigma = 1, L = 5$.}
\end{figure}

We implemented several baseline methods discussed thus far, and plotted their average error performance across $500$ iterations in Figure 1. For each iteration, we chose a signal $x$ randomly from the distribution $\mathcal{N}(0, I_L)$, as well as $N$ i.i.d noise vectors $\xi_i \sim \mathcal{N}(0, \sigma^2 I_L)$, and apply each of our methods. These simulations confirm our intuition that the UG-based SDP performs better than other benchmark methods. In particular, they suggest that the UG-based SDP is highly stable around integral instances. 

The implementation using bispectrum-like invariants is discussed in Appendix \ref{appendix_misc}. For each of the other procedures, we construct a $NL \times L$ matrix $W$ recording alignment preferences. 
For cross- or phase correlation, let $W_{ik;l}$ be the $(k-l)$th entry of the cross- or phase correlation vector between $y_1$ and $y_i$. For the spectral rounding off the Gram matrix $C$, and the solution of the semidefinite program $V$, let $W$ be the top $L$ eigenvectors of the respective matrix. The shifts are read off this matrix, and the un-shifted $y_i$ are averaged to produce an estimate for $x$. The first plot (a) shows the difference between this estimate and $x$. %

An important problem not very elaborated in this paper is how to determine the best shifts from $W$. A natural method is to identify an indicator vector $\mathds{1}\{ik: l_i \equiv k\}$ lying close to the column span of $W$, for some labelling $\boldsymbol{\ell} = (l_i)_i$. Currently, our procedure is to apply a linear transformation to $W$ such that the top $L \times L$ block of $W$ is the identity matrix $I_L$ (although there is no reason to believe this rounding is robust, it is sufficient for the purpose of having an easy benchmark between multiple methods). Then, for each $i$, the maximal entry in the first column of $W$ gives the choice of shift for $y_i$. The second plot shows the distance this indicator vector lies from the column span of $W$.

Let $W^{\perp} \in \mathbb{R}^{NL-L \times NL}$ be the matrix whose rows form the orthogonal complement of the rows of $W^T$, and $u \in \mathbb{R}^{NL}$ denote our indicator vector. Then $u$ should be a sparse vector near the nullspace of $W^{\perp}$. If we knew one of the non-zero coordinates of $u$, say coordinate $i$, we could recover the rest of $u$ as a sparse solution to $W^{\perp}_{-i} u_{-i} = -W^{\perp}_i$,
where $W^{\perp}_{-i}$ is the submatrix obtained from $W^{\perp}$ by removing
the column $W^{\perp}_i$. The problem can then be solved by
$\ell_1$-minimization, as suggested by the theory of sparse recovery. Although normally, we would be required to run this procedure for every coordinate (to make sure we pick an active one) the symmetries in $W$ imply the existence of $L$ such indicator variables and that any coordinate we pick will be active in one of them.
%
%

%% file: conclusion.tex
\section{Generalizations and Future Work}
\label{section_future}

It is worth noting that the discrete multireference alignment problem naturally generalizes from shifts over $\mathbb{Z}_L$ to actions of finite groups $G$ over finite spaces $S$. In this case, the analogue of phase correlation \cite{loneland10_phase} is defined in terms of the generalized Fourier transform $\mathcal{F}_{G,S}( \cdot, \rho) : \mathbb{C}^S \to \mathbb{C}^{d_\rho \times d_\rho}$, where $\rho: G \to \mathbb{C}^{d_\rho \times d_\rho}$ are irreducible matrix representations. The Unique Games SDP may also be generalized in a natural manner, in rough analogy with SDP variants studied for $\Gamma$-\np{MAX-2LIN} instances. In the case of an abelian group $G = S$, the proof of Theorem \ref{sdp_phasecorrelationequivalence} follows in the same fashion, and hence the alignment UG-based SDP without positivity constraints will remain equivalent to phase correlation. 

The symmetry of the SDP in this case can be described naturally by the theory of $\mathbb{C}*$-matrix algebras (refer to Appendix \ref{appendix_staralgebras}). 
Symmetries in semidefinite programs can be written with respect to linear combinations of $\{0,1\}$ matrices, which form a basis of a $*$-matrix algebra. Under sufficiently nice conditions, this basis can be block diagonalized, for example by Wedderburn's decomposition or by regular $*$-representations \cite{murota_blockdiagonal07}. Hence, the $*$-matrix algebra can be represented by a lower-dimensional block diagonalized version. From a computational perspective, this can make highly symmetric SDPs much more amenable to sparse SDP solvers \cite{dobre11_sdpsymmetry}. In our case, diagonalization by the block DFT allows the SDP to be rewritten as an optimization problem across $L$ PSD matrices of size $N \times N$ instead of one PSD matrix of size $NL \times NL$.

It would be interesting to see how the performance of the SDP changes when we study the alignment problem across more difficult groups, especially non-abelian ones, which arise in several applications.

The numerical simulations in Section \ref{section_numerics} suggest that the UG-based SDP achieves exact recovery with high probability for sufficiently high SNR. That is, the resulting SDP matrix is integral, so by solving the SDP we are indeed obtaining the solution to the quasi-ML estimator. Indeed, as the SNR decreases, there appears to be a phase transition during which the SDP 
almost always recovers an integral solution. Our analysis of the stability of the semidefinite program does not fully explain this phenomenon.
The authors believe this to be an interesting direction for future work, especially since guarantees of
exact recovery are attainable in high SNR settings for a few semidefinite relaxations, for example for the MIMO problem \cite{So_MIMO}.

Another important question is to understand the sample complexity of our approach to the alignment problem. Since the objective is to recover the underlying signal $x$, having a larger number of observations $N$ should give us better recovery. The question can be formalized as: for a given value of $L$ and $\sigma$, how large does $N$ need to be in order to have reasonably accurate recovery? The sample complexity of methods like the bispectral invariants would be expected to require $N = \Omega(\sigma^2 L^2 \log L)$ observations. We would hypothesize on the strength of our numerical results that the UG-based SDP requires fewer observations for meaningful recovery, 
and establishing this is an interesting open problem.

Along with expanding the domain of the alignment problem, it would be interesting to attempt the style of analysis discussed in this paper for other maximum likelihood problems. Maximum likelihood estimators play an important role in many estimation problems, but often (as in our problem) computing or approximating the MLE is a challenging problem and semidefinite programming could provide a tractable alternative in an average case setting.

%

%% file: appendix_misc.tex
\section{Phase correlation and the Bispectrum}
\label{appendix_misc}
%



Many scientific fields have isolated discussions and solutions for the alignment problem, occasionally with slight context-specific adjustments. Such methods include iterative template aligning \cite{kosir95_iterate}, zero phase representations \cite{zwart03_HRR}, angular synchronization \cite{sonday2011_sim}, and machine learning \cite{duch05_neural}. Unfortunately, most of these are either do not fairly weight each observation, do not make full use of the available information, or do not have rigorous performance discussion. We outline just a couple of very well-studied alignment methods in this section.

In the case of a pair of noisily shifted vectors ($N = 2$), several long-studied methods assign scores to each possible shift between the vectors, and estimate the shift as the one with the highest score. The most natural is to take the inner product between each possible shift of the vectors $y_i, y_j$, which gives the cross-correlation vector $v_{l}^{\np{cross}} = \langle y_i, R_{-l}y_j \rangle$, with maximal entry $\widehat{l} = \argmax v_l^{\np{cross}}$. 
Another frequently used in practice is the phase correlation vector \begin{align}
v^{\np{phase}} = \mathcal{F}^{*}\left\{\frac{\mathcal{F}(y_i,k) \mathcal{F}^*(y_j,k)}{|\mathcal{F}(y_i,k) \mathcal{F}^*(y_j,k)|}\right\}_{k=0}^{L-1}. \label{notation_phasecorrelation}
\end{align} For human-friendly images, phase correlation tends to be more appealing, since it tends to have a single sharp peak. The relation between cross- and phase correlation can be seen from the convolution theorem:
\begin{lemma}
\label{sdp_convolutionlemma}
Convolution theorem: 
$
\left\{\langle y_i, R_{-l} y_j) \rangle\right\}_{l=0}^{L-1} = \sqrt{L} \cdot \mathcal{F}\left\{\mathcal{F}(y_i,k) \mathcal{F}^*(y_j,k)\right\}_{k=0}^{L-1}.
$
\end{lemma}
\begin{proof}
Let $M_l$ denote the modulation operator $(M_lx)_k = x_k e(kl/L)$, defined such that $\mathcal{F} R_l x = M_l \mathcal{F} x$.
By Parseval's Theorem, \begin{align}
\left\{\langle y_i, R_{-l}y_j\rangle\right\}_l &= \left\{\langle \mathcal{F} y_i, M_{-l} \mathcal{F} y_j \rangle\right\}_l 
= \bigg\{\sum_{k \in \mathbb{Z}_L} e\left(\tfrac{-kl}{L}\right) \mathcal{F}(y_i,k) \mathcal{F}^*(y_j,k)\bigg\}_l \notag \\
&= \sqrt{L} \cdot \mathcal{F}\left\{\mathcal{F}(y_i,k) \mathcal{F}^*(y_j,k)\right\}_k. \label{sdp_phasecorrelation_trick}
\end{align}
\end{proof}

Cross- and phase correlation can be used directly for the problem of multireference alignment, say by aligning all of the observations against the first one. However, this is certainly not a robust method. 

Another relevant notion for characterizing alignments are the moment spectra. These are properties of a signal which are invariant under shifts. The $k$th power spectrum of a real signal $x \in \mathbb{R}^L$ is defined by $|\mathcal{F}(x,k)|^2 = \mathcal{F}(x,k)\mathcal{F}^*(x,k)$. This can be extended to the $d$-th order bispectrum (or moment spectra) \begin{align}
b(k_1,k_2, \ldots, k_d) = \mathcal{F}(x,k_1 + k_2 + \cdots + k_d) \mathcal{F}^*(x,k_1) \mathcal{F}^*(x,k_2) \cdots \mathcal{F}^*(x,k_d);
\end{align} its shift invariance can be seen since it is the Fourier transform of the $d$-th order autocorrelation function \begin{align}
a(k_1,k_2, \ldots, k_d) = \sum_{l=0}^{L-1} x_{l} \overline{x_{l-k_1} x_{l-k_2} \cdots x_{l-k_d}},
\end{align} by the Wiener-Khinchin theorem.
How much information do these invariants capture about the signal $x$? Sadler and Giannakis give iterative and least squares algorithms in \cite{sadler92_bispectrum} for reconstructing the Fourier phases from full knowledge of the bispectrum, and Kakarala in \cite{kakarala93_uniqueness} show the uniqueness of real function given all of its higher order moment spectra. These arguments generally tend to be of a more number-theoretic flavor and may be tedious in practice. 

A simple special case occurs when we take $k_1 = k_2 = \ldots = k_d = 1$ for $d = 1, \ldots, L$. This gives us the set of bispectral invariants $\mathcal{F}(x,d) \mathcal{F}^*(x,1)^d$. This quantity, as a shift invariant, can be consistently estimated from our observations $y_i$. With an estimate for the phase of the first Fourier coefficient $\mathcal{F}(x,1)$ (for example we can sample over a discretization of the unit circle), we can recover the remaining Fourier phases $\mathcal{F}(x,d)$ and thus the signal $x$.


\section{Maximum Likelihood}


\begin{lemma}
\label{appendix_mle_gram}
The data Gram matrix $C$ with entries $C_{ik;jl} = \left\langle R_{-k} y_i, R_{-l} y_j \right\rangle$ satisfies the following properties:
\begin{enumerate}
\item $C \succeq 0$ and has rank $L$, with non-zero eigenvalues $\lambda_k = L \sum_{i=1}^N \left|\mathcal{F}(y_i,k)\right|^2.$
\item There is a unitary matrix $\mathcal{P}$ for which $\mathcal{P}C\mathcal{P}^* = \diag(\mathcal{C}_0, \ldots, \mathcal{C}_{L-1})$ is block diagonal, where each $\mathcal{C}_k \in \mathbb{C}^{N \times N}$.
\end{enumerate}
\end{lemma}
\begin{proof}
$C$ has Cholesky decomposition $$C = (R\!y)(R\!y)^T \in \mathbb{R}^{LN \times LN} \quad \text{where } R\!y = \begin{pmatrix} \vdots \\ - \ R_{-k}y_i \ - \\ \vdots \end{pmatrix}_{i,k},$$ so $C \succeq 0$ has rank $L$. Since $\langle R_{-k}y_i, R_{-l}y_j \rangle = \langle R_{-k+m}y_i, R_{-l+m}y_j \rangle$, $C$ is composed of $N \times N$ circulant blocks of size $L \times L$. After permuting the rows and columns of $C$, one may write $C$ as a block circulant matrix $$P^{T}CP = \begin{pmatrix} C_0 & C_1 & \cdots & C_{L-1} \\ C_{L-1} & C_{0} & \cdots & C_{L-2} \\ \vdots & \vdots & \ddots & \vdots \\ C_1 & C_2 & \cdots & C_0 \end{pmatrix} \quad \text{where }(C_k)_{ij} = \langle y_i, R_{-k} y_j \rangle,$$ and $P$ is the appropriate permutation matrix. 
$P^TCP$ is block diagonalized by the block Discrete Fourier Transform matrix $DFT_L \otimes I_N$ \cite{ligong04_blockcirculant}, where $\otimes$ denotes the Kronecker product. $\mathcal{P} = (DFT_L \otimes I_N)P^T \in \mathbb{C}^{NL \times NL}$ is thus a unitary transformation such that \begin{align}\mathcal{P}C \mathcal{P}^* = \diag(\mathcal{C}_0, \ldots, \mathcal{C}_{L-1}); \quad \mathcal{C}_k \in \mathbb{C}^{N \times N} 
\end{align} is block diagonal. These block diagonals are componentwise Discrete Fourier Transforms of the ``vector'' $\left(C_0, C_1, \ldots, C_{L-1}\right)$: $$\mathcal{C}_k = \sum_{l = 0}^{L-1} e\left(\tfrac{kl}L\right) \, C_l = \sqrt{L} \cdot \left[\mathcal{F}^*(C_0, C_1, \ldots, C_{L-1})\right]_k = \left\{ L (Y_i)_k (Y_j^*)_k\right\}_{ij}$$ by Lemma \ref{sdp_convolutionlemma}. Note also that each $\mathcal{C}_k$ is Hermitian and rank $1$ ($C$ is of rank $L$, and each of the blocks $\mathcal{C}_k$ has positive rank). The unique nonzero eigenvalue $\lambda_k$ of $\mathcal{C}_k$ is given by
\begin{align}
\lambda_{k} &= \tr(\mathcal{C}_k) = L \sum_{i=1}^N \left|(Y_i)_k\right|^2.
\end{align} 
\end{proof}

%% file: appendix_nphard.tex
\section{NP-hardness of Quasi MLE}
\label{section_appendix_nphard}

%

\begin{theorem}\label{mle_alignmentproblem}
Consider the problem $\mathsf{ALIGN}(y_1, \ldots, y_N)$: for vectors $y_1,\dots,y_N \in \mathbb{R}^L$, find the shifts $\boldsymbol{\ell} = (l_1,\dots,l_N)$ which maximize
\[
\mathcal{A}(l_1, \ldots, l_N) = \sum_{i,j \in [N]} \langle R_{-l_i} y_i, R_{-l_i}y_j \rangle.
\] Let $\mathcal{A}^* = \max_{\boldsymbol{\ell}} \mathcal{A}(\boldsymbol{\ell})$. It is \np{NP}-hard (under randomized reductions) to esimate $\mathcal{A}^*$ within $\frac{16}{17} + \varepsilon$ of its true value. It is \np{UG}-hard (under randomized reductions) to estimate 
$\mathcal{A}^*$ within any constant factor. 
\end{theorem}

We demonstrate this by a poly-time approximation preserving reduction from a special class of \np{MAX-2LIN}$(q)$ instances called $\Gamma$-\np{MAX-2LIN}$(q)$.
Consider a (connected) \np{MAX-2LIN}$(q)$ instance where each constraint has the form $x_i - x_j \equiv c_{ij} \pmod{q}$, of which at most $\rho^{*}$ are satisfiable. Representing each variable as a vertex of a graph and each constraint as an edge, the instance is associated with a connected graph $G = (V(G), E(G)),$ where $V(G) = [N]$ and $|E(G)| = M$. 

Let $\np{poly}(M)$ be the space of integer functions which are bounded by polynomial order, i.e. $f \in \np{poly}(M)$ iff there are constants $C,k$ such that $f(M) \le CM^k$ for all $M \ge 1$. We say that an event occurs w.h.p if it occurs with probability $1-\epsilon(q,G)$, where 
$\frac{1}{\epsilon(q,G)} \notin \np{poly}(q \cdot |E(G)|)$. Notice under this definition that if $\np{poly}(qM)$ events occur w.h.p, then by an union bound the event that all occur will also be w.h.p.

Construct a parameter $\templen = \templen(q,M) \in \np{poly}(qM)$. We take the vectors $y_1, \ldots, y_N$ to be of length $L = qM\templen$. The indices of the vector $y_i$ can be expressed in mixed radix notation by elements of the tuple $(\mathbb{Z}_q, E(G), [\templen])$. For example, we would associate the tuple index $(x, e_k, t)$ of $y_i$ by the index $x \cdot qM + e_k \cdot M + t$, where $e_k$ is the $k$th edge in some enumeration of $E(G)$. Note that a shift by $c \cdot qM$ takes this tuple index to $(x, e_k, t) + c \cdot qM \to (x + c, e_k, t)$. 

For each edge constraint $x_i - x_j \equiv c_{ij}$, let $z_{ij}$ be a vector uniformly at random chosen from $\{\pm 1\}^\templen$. Assign the entries $(0, \{i,j\}, \cdot)$ of $y_i$ to $z_{ij}$, and the entries $(c_{ij}, \{i,j\}, \cdot)$ of $y_j$ to $z_{ij}$. Assign all of the remaining entries of the $y_i$'s to i.i.d Rademacher random variables ($\{\pm 1\}$ with probability $1/2$). Intuitively, a relative shift of $c_{ij} \cdot qM$ between $y_i$ and $y_j$ will produce a large inner product due to the overlapping of the $z_{ij}$'s, while any other shift between them would produce low inner products.

\begin{lemma}\label{mle_hoeffding}
Suppose $\gamma \in \mathsf{poly}(qM)$. Consider two random vectors $y_1, y_2$ of length $\gamma$ whose entries are i.i.d sampled from the Rademacher distribution. W.h.p, for any $0 < \epsilon \ll 1$, the inner product of every possible shift $R_\ell y_1$ of $y_1$ with $y_2$ is bounded in absolute value by $\sqrt{\gamma} \cdot m^\epsilon$.
\end{lemma}
\begin{proof}
By independence, each inner product is the sum of $\gamma$ independent Bernoulli random variables which take values $\pm 1$ with probability $1/2$. Hoeffding's inequality indicates
\[
 \prob(\langle R_\ell y_1, y_2 \rangle \geq \sqrt{\gamma} \cdot (qM)^\epsilon) \leq 2\exp\left\{ - (qM)^{\epsilon}/2 \right\}, 
\] so $\frac{1}{\prob(\langle R_\ell y_1, y_2 \rangle \geq \sqrt{\gamma} (qM)^\epsilon)} \notin \mathsf{poly}(qM)$. Union bounding over all $\gamma = \mathsf{poly}(qM)$ such inner products, w.h.p all of the inner products are simultaneously bounded by $\sqrt{\gamma} \cdot (qM)^\epsilon$.
\end{proof}

We say an edge $\{i,j\} \in E(G)$ is $c_{ij}$-satisfied under a labelling $\boldsymbol{\ell}$ if $l_i-l_j \equiv c_{ij}qM \pmod{L}$. From Lemma \ref{mle_hoeffding}, we observe that w.h.p, for any choices of shifts $l_i, l_j$, 
\[
|\langle R_{-l_i} y_i,\ R_{-l_j}y_j \rangle - \templen \cdot \delta(\{i,j\} \in E(G)\text{ is }c_{ij}\text{-satisfied})| < (qM)^\epsilon \sqrt{L}.
\] 
It follows that if the labelling $\boldsymbol{\ell}$ induces exactly $k$ $c_{ij}$-satisfied edges, 
w.h.p
\begin{align}
\label{mle_objectivealphashift} \big|\mathcal{A}(\boldsymbol{\ell}) - 2k\templen\big| < k (qM)^\epsilon \sqrt{L}  + \left(N^2 - k\right) (qM)^\epsilon \sqrt{L} \le q^\epsilon M^{2+\epsilon} \sqrt{L}. 
\end{align} 
\begin{theorem}
\label{mle_labelstocut}
From a given labelling $\boldsymbol{\ell}$, w.h.p one may in polynomial time construct $(x_1, \ldots, x_N) \in \mathbb{Z}_q^N$ satisfying at least $\big(\mathcal{A}(\ell) - q^\epsilon M^{2+\epsilon}\sqrt{L}\,\big)/(2\templen)$ edge-constraints $x_i - x_j \equiv c_{ij}$. Conversely, w.h.p there is a labelling $\boldsymbol{\ell}^{\np{max}}$ w.h.p satisfying $\mathcal{A}\left(\boldsymbol{\ell}^{\np{max}}\right) > 2\templen \cdot \rho^{*} M - q^\epsilon M^{2+\epsilon} \sqrt{L}$.
\end{theorem}
\begin{proof}
Consider the subgraph $H$ of $G$ with vertex set $V(G)$ and edge set comprising all edges $c_{ij}$-satisfied under $\boldsymbol{\ell}$. 
For each connected component of $H$, arbitrarily choose a vertex $i$ of the component to have $x_i = 0$. Follow a spanning tree of each connected component and assign each neighbor $j$ by $x_j \equiv x_i - c_{ij} \pmod{q}$. 
The first implication of the lemma follows immediately by applying 
(\ref{mle_objectivealphashift}).
Conversely, construct the labelling $\boldsymbol{\ell}^{\np{max}}$ by setting $l^{\np{max}}_i = x_i \cdot qM$. This labelling induces at least $\rho^{*}M$ $c_{ij}$-satisfied edges, and (\ref{mle_objectivealphashift}) completes the lemma. 
\end{proof}

\begin{corollary}
Take $\templen > q^{2}M^{4}$.
If $\boldsymbol{\ell}^{\np{max}}$ satisfies $\mathcal{A}(\boldsymbol{\ell}^{\np{max}}) \ge \left(\alpha + \delta\right) \mathcal{A}^*$, then w.h.p, from $\boldsymbol{\ell}^{\np{max}}$ one may construct $(x_1, \ldots, x_N) \in \mathbb{Z}_q^N$ in poly-time satisfying $(\alpha + \delta')\rho^*$ fraction of $\mathbb{Z}_q$\np{-MAX-2LIN} constraints, for some $\delta' > 0$.
\end{corollary}
\begin{proof}
Let $\rho$ be the fraction of variable assignments satisfied by the construction of Theorem \ref{mle_labelstocut} for the labelling $\boldsymbol{\ell}^{\np{max}}$. W.h.p, \[
\rho > \frac{\mathcal{A}(\boldsymbol{\ell}^{\np{max}}) - q^\epsilon M^{2 + \varepsilon} \sqrt{L}}{2\templen M} > 
\frac{(\alpha + \delta)\mathcal{A}^* 
- q^{1/2+\epsilon} M^{5/2 + \epsilon} \templen^{1/2}}{2\templen M} > (\alpha + \delta) \rho^{*} - \frac{q^{1/2+\epsilon} M^{3/2 + \epsilon}}{\templen^{1/2}}.
\] 
\end{proof}

 
Assuming the Unique Games conjecture, it is NP-hard to approximate $\Gamma$-\np{MAX-2LIN}$(q)$ to any constant factor \cite{KKMO}.
Hence it is UG-hard (under randomized reductions) to approximate {\sf ALIGN} within
any constant factor.
\textsf{MAX-CUT} is a special case of $\Gamma$-\np{MAX-2LIN}$(q)$.
Since it is \textsf{NP}-hard to approximate \textsf{MAX-CUT} within 
$\frac{16}{17} + \varepsilon$, it is  \textsf{NP}-hard (under randomized reductions)
to approximate \textsf{ALIGN} within $\frac{16}{17} + \varepsilon$.


%% file: appendix_staralgebra.tex
\section{$*$-matrix algebras} 
\label{appendix_staralgebras}
In this appendix section, we briefly sketch the ideas behind symmetry reduction in semidefinite programs. Most of the details are deferred to other sources. For example, \cite{dobre11_sdpsymmetry} is a good reference on this topic. While the application of symmetry to the SDPs we are interested in is relatively straightforward and can be explained without reference to $*$-matrix algebras, the notation and ideas behind the theory of $*$-matrix algebras make it especially clear to express.

\begin{definition}
A complex (or real) matrix $*$-algebra $\mathcal{A}$ is a set of $n \times n$ matrices in $\mathbb{C}^{n \times n}$ (or $\mathbb{R}^{n \times n}$) such that for any $X, Y \in \mathcal{A}$,
\[
\alpha X + \beta Y,\ X^*,\ XY \in \mathcal{A} \quad \text{for all } \alpha, \beta \in \mathbb{C}\ (\text{or }\mathbb{R}).
\] 
A coherent configuration is a set of matrices $\{A_1, \ldots, A_d\} \subset \{0,1\}^{n \times n}$ satisfying:
\begin{itemize}
\item Some subset of the $A_i$'s sum to $I_n$.
\item $\sum_{i=1}^d A_i = \bold{1}_{n \times n}.$
\item $A_i^T \in \mathcal{A}$.
\item $A_i A_j$ is a linear combination of $A_1, \ldots, A_d$.
\end{itemize} The span of the $A_i$'s is a $*$-matrix algebra $\mathcal{A}$. If the $A_i$'s commute and contain the identity matrix $I_n$, then this configuration is called an association scheme, and $\mathcal{A}$ is also known as a Bose-Mesner algebra. Treating $\mathcal{A}$ as a matrix vector space, a coherent configuration forms a basis for $\mathcal{A}$. 
\end{definition}

For example, $\mathbb{C}^{n \times n}$ is a matrix $*$-algebra. 
Another commonly considered matrix $*$-algebra is the space of $G$-invariant matrices. Let $S$ be a finite set, and let $G$ be a finite group acting on $S$. The space of $G$-invariant matrices are the matrices $A \in \mathbb{C}^{|S| \times |S|}$ for which $P_g^TAP_g = A$, where $P_g$ is the permutation matrix associated with $g \in G$. The space of $G$-invariant matrices forms a matrix $*$-algebra generated by the coherent configuration $\{P_g\}_{g \in G}$. In particular, the basis is closed under multiplication, since $P_{g_1}P_{g_2} = P_{g_2 \circ g_1}$.

\begin{definition}
A $*$-homomorphism $\phi$ is a linear mapping of $*$-matrix algebras $\phi:\mathcal{A} \to \mathcal{B}$ which additionally satisfies $\phi(A^*) = \phi(A)^*$, $\phi(AB) = \phi(A)\phi(B)$, and $\phi(I_{\mathcal{A}}) = I_{\mathcal{B}}$ if the identity matrix $I_{\mathcal{A}} \in \mathcal{A}$.
\end{definition} 
\begin{lemma} 
\label{background_starhomomorphism}
If $\phi$ is an injective $*$-homomorphism $\phi: \mathcal{A} \to \mathcal{B}$, then $A \in \mathcal{A}$ and $\phi(A)$ share the same set of eigenvalues (ignoring multiplicity), and $A \succeq 0 \Longleftrightarrow \phi(A) \succeq 0$.
\end{lemma}

To leverage symmetry in semidefinite programs, it will be useful to take the images of the objective semidefinite matrix under dimension-reducing $*$-homomorphisms. Two such $*$-homomorphisms which may be systematically constructed are given by Artin-Wedderburn's Theorem and regular $*$-representations \cite{murota_blockdiagonal07}. Consider a primal semidefinite problem of the form \begin{align}
\text{minimize}& \quad \tr(C_0V) \label{background_sdp_symmetryreduction} \\
\text{subject to}& \quad V \succeq 0 \quad \text{and}\quad \tr(C_iV) = b_i, \notag \end{align} where $V$ is the semidefinite matrix to optimize over, and $C_i \in \mathbb{R}^{n \times n}$ are data matrix constraints and $b = (b_1, \ldots, b_m)$ are given, for $i \in \{1,2, \ldots, m\}$. 
Suppose the data matrices $C_i$ are contained in a low-dimensional complex matrix $*$-algebra $\mathcal{A}_{SDP}$. 
\begin{theorem}
Suppose the data matrices $\{C_i\}_{i=0}^m \subset \mathbb{R}^{n \times n}$ are real symmetric. Let $\mathcal{A}_{SDP}$ be a complex $*$-algebra with a real orthogonal basis containing all of the $C_i$'s. Then 
\begin{enumerate}
\item there is an optimal solution $V$ to the primal SDP problem contained in $\mathcal{A}_{SDP}$ 
\item $\text{Re}(V) \in \mathcal{A}_{SDP}$ is also an optimal solution to the primal SDP problem.
\end{enumerate}
\end{theorem}
\begin{proof}
Refer to (\cite{dobre11_sdpsymmetry}, Corollary 2.5.2).
\end{proof}

A special case of this is the result applies for the $*$-algebra of $G$-invariant matrices.
\begin{lemma} \label{appendix_staralgebra_groupinvariant}
In the semidefinite program (\ref{background_sdp_symmetryreduction}), suppose the constraint matrices $C_0, C_1, \ldots, C_m$ are $G$-invariant. Then there is a solution $V$ to the SDP which is also $G$-invariant.
\end{lemma}
\begin{proof}
If $V'$ is a solution of (\ref{background_sdp_symmetryreduction}), then the matrix $V \in \mathbb{C}^{n \times n}$ given by \begin{align*}
V = \mathcal{R}_G(V) = \frac{1}{|G|} \sum_{g \in G} P_g^T V' P_g
\end{align*} is $G$-invariant, feasible, and has the same objective value as $V'$. The averaging map $\mathcal{R}_G$ is known as the \textit{Reynolds operator}.
\end{proof}

%% file: appendix_stability.tex
\section{Stability}
\label{appendix_stability}



\begin{lemma}
\label{stability_tailbounds_eta}
$\mathbb{E}[\eta_{ij}|\xi_i] \le \mu_\eta \|\xi_i\|$ and $\mathbb{E}[\eta_{ij}^2|\xi_i] \le \sigma_\eta \|\xi_i\|^2$, where $\mu_\eta = 2\sigma(\log L + 1)$ and $\sigma_\eta^2 = 4\sigma^2(\log
^2 L + 1)$. 
\end{lemma}
\begin{proof}
Take $\zeta \sim \mathcal{N}(0, I_L)$, and choose a unit vector $z \in \mathbb{R}^L$. 
Define $\eta(z) = \max_{l} \eta_l(z)$, where the distribution of $\eta_l(z) = |\langle R_l z, \zeta\rangle|$ is independent of the choice of $z$.
Union bounding,
\[
 \prob\left( \eta(z) \geq t \right) \leq \sum_{l \in \mathbb{Z}_L} \prob\left(\eta_l(z) \geq \tfrac t2 \right) \le \sum_{l} \prob\left(|\langle e_l, \zeta \rangle| \ge \tfrac t2 \right) = \sum_{l} \prob\left( |\zeta_l| \geq \tfrac t2 \right) \le \min\bigg\{1, \frac{2L}t \cdot \frac{e^{-t^2/2}}{\sqrt{2\pi}}\bigg\}.
\]
where $e_l \in \mathbb{R}^L$ is a standard unit basis vector. Thus, for $T > 1$,
\begin{align*}
\mu_\eta &= \mathbb{E} \, \eta(z) =\int_0^{\infty} \prob\left( \eta(z) \geq t \right) dt \leq T + \frac{2L}{\sqrt{2\pi}} \int_T^{\infty} te^{-t^2/2} \, dt \le T + L e^{-T} 
\end{align*}
 and
\begin{align*}
\sigma_\eta^2 &= \mathbb{E} \, \eta(z)^2 = \int_0^{\infty} 2t \prob\left( \eta(z) \geq t \right) dt \leq T^2 + \frac{L}{\sqrt{2\pi}} \int_T^{\infty} te^{-t^2/2} \, dt \le T^2 + L e^{-T}.
\end{align*}
The lemma follows by taking $T = \log L$, since $\eta_{ij} = 2\sigma\|\xi_i\| \cdot \eta\big(\xi_i / \|\xi_i\|\big)$ conditional on $\xi_i$.
\end{proof}

\begin{lemma}
\label{stability_independenttailbounds}
Let $t > 0$. With probability at least $1 - 2 N e^{-t^2 N}$, 
\begin{align}
&\sum_{i \neq j} (\eta_{ij} + \eta_{j0}) D_{ij} \le (\mu_\eta + \sigma_\eta t)\bigg(2\|x\| + \frac 1N\sum_i \|\xi_i\|\bigg) N^2.
\label{stability_independenttailbounds_ineq}
\end{align}
holds for all $D_{ij} \in [0,1]$.
\end{lemma}
\begin{proof} 
Define the random variables $Z_{ij} = (\eta_{ij} + \eta_{j0}) D_{ij} \le \eta_{ij} + \eta_{j0}$. 
For fixed $i$, the $\eta_{ij}$'s are independent positive random variables.
The one-sided tail bound for independent positive random variables by Maurer \cite{maurer2003_bound} gives \[
\prob\bigg(\sum_j (\eta_{ij} + \eta_{j0}) - \sum_j \mathbb{E} \, (\eta_{ij} + \eta_{j0}) \ge t_i \sqrt{\sum_j \mathbb{E} \, (\eta_{ij}^2 + \eta_{j0}^2)} \, \bigg) \le 2e^{-t^2_i}.
\] Choose $t_i = t\sqrt{N}$.
By union bound, with probability at least $1 - 2N e^{-t^2N}$,
\begin{align*}
\sum_{i,j} Z_{ij} &\le \sum_{i} \mu_\eta (2\|x\| + \|\xi_i\|) N + \sum_i t \sqrt{N^2 \sigma_\eta^2 (4\|x\|^2+ \|\xi_i\|^2)} \\
&\le (\mu_\eta + \sigma_\eta t) \bigg(2\|x\| + \frac 1N \sum_i \|\xi_i\|\bigg) N^2.
\end{align*}
\end{proof}

\begin{theorem}
\label{appendix_stability_theorem}
With probability $1 - e^{-N + o(N)}$, the solution to the SDP satisfies $$
\sum_{i,j} D_{ij} \le \frac{(\|x\| + \sigma^2 \sqrt{L}) \cdot 12\log eL }{\Delta} \cdot N^2.
$$
\end{theorem}
\begin{proof}
Let $V$ be a solution matrix to the SDP, so that $\tr(CV) \ge \tr(CV^{int})$. It follows from Lemma \ref{appendix_stability_ineq1} that $$ \Delta \sum_{ij} D_{ij} \le \sum_{ij} (\eta_{ij} + \eta_{j0}) D_{ij}.$$
The $\chi^2$-tail bound of Laurent-Massart \cite{Laurent_Massart_TailChiSquare} states
\begin{align}
\prob \bigg(\frac{1}{\sigma^2}\sum_i \|\xi_i\|^2 \le NL + 2\sqrt{NL t} + 2t\bigg) \ge 1-\exp\{-t\}. \label{massartlaurent}
\end{align} When $t = NL$ this gives \[
\sum_i \|\xi_i\| \le \sqrt{N \sum_i \|\xi_i\|^2} \le \sigma N\sqrt{5L}.
\] with probability at least $1 - e^{-NL}$. Union bounding with the tail bound of Lemma \ref{stability_independenttailbounds}, and applying Lemma \ref{stability_tailbounds_eta}, $$\Delta \sum_{ij} D_{ij} \le (\mu_\eta + \sigma_\eta)\bigg(2\|x\| + \frac 1N\sum_i \|\xi_i\|\bigg) N^2 \le 12 \log eL(\|x\| + \sigma \sqrt{L}) N^2$$ with probability at least $1 - e^{-N + o(N)}$.
\end{proof}

\begin{lemma}
\label{appendix_stability_johnsonlindenstrauss}
Let $0 < \delta, \varepsilon, \varepsilon' \ll 1$ be small constants, and define $\alpha_{JL}(x) = (1+\delta) x + \varepsilon$. There is a set of unit vectors $\mathfrak{N}$ of size at most $$\exp\left(\mathcal{O}\left(\delta^{-2} \log(1/\varepsilon) \log(1/\varepsilon') \right)\right)$$ such that for any set of unit vectors $\{v_i\}$, there is a map $\varphi : \{v_i\} \to \mathfrak{N}$ satisfying the inequality $$\alpha_{JL}^{-1}\bigg(\|v_i-v_j\|^2\bigg) \le \|\varphi(v_i) - \varphi(v_j)\|^2 \le \alpha_{JL}\bigg(\|v_i-v_j\|^2\bigg)$$ for at least $1 - \varepsilon'$ fraction of the pairs $(i,j) \in [N] \times [N]$.
\end{lemma}
\begin{proof}
This lemma appears frequently in SDP literature, and in particular is used to analyze adversarial semi-random \textsf{Unique-Games} instances in \cite{Kolla_Makarychev2_UG}. Notice that the size of the set $\mathfrak{N}$ is independent of the number of vectors in the set $\{v_i\}$. 
 
Construct a $\varepsilon/32$-net $\mathfrak{N}$ of the unit hypersphere in a $\mathcal{O}(\delta^2 \log(1 / \varepsilon'))$-dimensional space $\mathfrak{L}$. By the (strong version of the) Johnson-Lindenstrauss lemma, there is a randomized mapping $\varphi': \{v_i\} \to \mathfrak{L}$ satisfying $$(1-\delta/2)\|v_i-v_j\|^2 \le \|\varphi'(v_i) - \varphi'(v_j)\|^2 \le (1+\delta/2)\|v_i-v_j\|^2$$ with probability at least $1 - \varepsilon'$. Define $\varphi(v_i)$ to be the closest point to $\varphi'(v_i)$ in $\mathfrak{N}$, and observe that
$$(1-\delta/2)x - \varepsilon \ge \frac{x - \varepsilon}{1 + \delta} = \alpha_{JL}^{-1}(x),\quad (1 + \delta/2)x + \varepsilon \le \alpha_{JL}(x)$$ for all $x > 0$, so $\varphi$ satisfies the conditions of the lemma.
\end{proof}